\documentclass[journal,twoside]{IEEEtran}

\usepackage{amsmath}
\usepackage{amssymb}
\usepackage{amsfonts}
\usepackage{amsthm}

\usepackage[T1]{fontenc}
\usepackage{txfonts}

\usepackage[svgnames]{xcolor}
\usepackage{subfigure}
\usepackage{tikz}

\usepackage{acronym}
\usepackage{latexsym}
\usepackage{xspace}

\newcommand{\B}{\varmathbb{B}}

\newcommand{\R}{\varmathbb{R}}

\DeclareMathOperator{\ex}{\varmathbb{E}}

\DeclareMathOperator{\prob}{\varmathbb{P}}

\theoremstyle{plain}
\newtheorem{theorem}{Theorem}

\newtheorem*{corollary*}{Corollary}

\theoremstyle{definition}

\newtheorem*{definition*}{Definition}

\theoremstyle{remark}

\newtheorem*{remark*}{Remark}

\newcommand{\bM}{\mathbf{M}}

\newcommand{\bh}{\mathbf{h}}

\newcommand{\br}{\mathbf{r}}

\newcommand{\bx}{\mathbf{x}}
\newcommand{\by}{\mathbf{y}}
\newcommand{\bz}{\mathbf{z}}

\def\b0{0}

\newfont{\bb}{msbm10 scaled 1100}

\begin{document}
\bstctlcite{BSTcontrol}


\title{Effects of Mobility on User Energy Consumption and Total Throughput in a Massive MIMO System}

\author{%
Aris~L.~Moustakas$^{\star \dagger}$
,
Luca Sanguinetti$^{\ddagger \dagger} $
and
M{\'e}rouane Debbah$^{\dagger}$
\thanks{%
$^\star$: \small{Physics Department of the University of Athens, Greece};

$^\dagger$: \small{Alcatel-Lucent Chair, Ecole sup{\'e}rieure d'{\'e}lectricit{\'e} (Sup{\'e}lec), Gif-sur-Yvette, France};

$^\ddagger$ \small{Dipartimento di Ingegneria dell'Informazione, University of Pisa, Pisa, Italy}
}
}


\maketitle

\begin{abstract}
Macroscopic mobility of wireless users is important to determine the performance and energy efficiency of a wireless network, because of the temporal correlations it introduces in the consumed power and throughput. In this work we introduce a methodology that obtains the long time statistics of such metrics in a network. After describing the general approach, we present a specific example of the uplink channel of a mobile user in the vicinity of a massive MIMO base-station antenna array. To guarantee a fixed SINR and rate, the user inverts the path-loss channel power, while moving around in the cell. To calculate the long time distribution of the consumed energy of the user, we assume his movement follows a Brownian motion, and then map the problem to the solution of the minimum eigenvalue of a partial differential equation, which can be solved either analytically, or numerically very fast.  We also treat the throughput of a single user. We then discuss the results and how they can be generalized if the mobility is assumed to be a Levy random walk. We also provide a roadmap to use this technique when one considers multiple users and base stations.
\end{abstract}


\section{Introduction}
\label{sec:Introduction}

Perhaps the most challenging property of the wireless propagation channel is its temporal variability. Since the first mobile telephones appeared at a massive scale, engineers had to address the temporal fluctuations of the received signal power to make sure that a call connection remained active. Adverse effects that needed to be countered include (a) fading holes of the channel due to multiple wave reflections, (b) change in the link distance due to physical movement of the wireless device away from the base, or (c) interference fluctuations due to movement of interfering devices.

To understand the behavior of fading various models were proposed with varying complexity, starting from the  Jakes model\cite{Jakes_book} to more involved correlated models in both frequency and time \cite{Calcev2004_3GPP_SCM}. The introduction of concise fading models has made it possible to obtain, together with numerical simulations, analytical expressions for the quantities of interest, such as ergodic and outage capacities etc\cite{Tulino2004_RMT_review}. This in turn allowed the introduction of ways to counter the adverse effects of fading, through scheduling and space-time coding.

The effects of path-loss have also been studied more recently in a large body of work using the theory of Poisson point processes  (PPP) \cite{Haenggi2008_InterferenceLargeWirelessNets, Baccelli2009_StocGeometryWirelessNetsTheory, Baccelli2009_StocGeometryWirelessNetsApplications}. This approach has provided a good understanding of the effects of randomness in position of mobile devices in a network and has allowed a thorough characterization of the statistics of interference in wireless networks. However, most of the analysis deals with static or near-static networks and does not take into account the consequences of macroscopic mobility of users.

This shortcoming is important when one realizes that due to not too rapid mobility there are temporal correlations in the necessary power for any given user. For example an untypically  high density of users at the cell edge will result to an increased energy consumption over an extended period of time, which may drain the available energy resources of a base-station. This in turn becomes important especially for off-grid deployments, with finite energy resources. Therefore, it is important to quantify not only how often such unlikely events happen, but also how long they last, which depends on the user mobility.

To analyze the effects of mobility, simple yet effective models that describe the statistics of humans moving around are necessary. Several models have been proposed\cite{Camp2002_SurveyMobilityModelsAdHocNetworks, Fan2004_SurveyMobilityModels} and their statistical properties have been studied \cite{LeBoudec2005_StationarityMobilityModels} in detail and have been implemented in numerical simulators. In particular, three types of mobility models are more popular. The first and simplest one is a random walk (RW). This model is a continuous time Markov process on a lattice with a step size distribution which has zero mean and finite variance. At long times and distances this can be approximated by a Brownian motion (BM). Another, more involved one, corresponds to a Markov process (LRW) with infinite variance in the step size distribution, due to the long tails in the step sizes, corresponding to Levy processes. This has been proposed as a more realistic model for human mobility\cite{Rhee2011_LevyHumanMobility, Scafetta2011_LevyWalkHumanMobility}. Finally, the so-called random waypoint process (RWP) has also been proposed, in which a mobile user picks a random destination and travels with constant velocity to reach it, which is also a Markov process. However, despite the well-understood properties of the mobility models above not much progress \cite{Galluccio2010_OpportunisticCommunicationsInfostations}
has been made towards providing analytical results for the long term statistics of communications performance metrics, such energy consumed or total throughput.

In this paper we take advantage of the Markovian property of user mobility to analyze the long time statistics of these performance measures in a network with mobility. We believe that the approach is fairly general to encompass all Markovian mobility models described above, at least in principle. The methodology is based on a simple, but powerful theorem, the so-called Feynman-Kac formula\cite{Simon1979_FunctionalIntegrationQuantumPhysics_book, Pastur_book1992_SpectraRandomOperators}, which maps the average over all random walks to the minimum eigenvalue of an partial differential equation, which is related to the random walk.

\subsection{Outline}
\label{sec:Outline}

In the next section we will define the metrics of interest, in terms of an integral over time. We will introduce the mobility model and show how we can calculate the probability of finding a user at a given location using the diffusion equation. In Section \ref{sec:math_formalism} we state the main result of the paper, namely Theorem \ref{thm:MainResult} and outline its proof. In Section \ref{sec:EnergyStatistics} we calculate the long time statistics of the uplink  energy consumption, while in Section \ref{sec:Throughput Statistics} we do the same for the total throughput. Finally, in Section \ref{sec:discussion} we discuss how these methods can be generalized to other situations.

\section{Problem Statement and Setup}
\label{sec:problem_statement}

The purpose of this paper is to present a methodology to analyze the long time $T$ statistics of quantities of the form
\begin{equation}\label{eq:energy_def}
E_T = \int_0^T dt\, V(\br(t))
\end{equation}
where $V(\bx)$ is a function of the position of one or more mobile devices in a network and the devices move around the region of interest over time. We will assume here that $V(\bx)\geq 0$ and that it is bounded from above.

\subsection{Metrics of Interest}
\label{sec:metrics}

$V(\bx)$ is a function that can represent a number of relevant metrics of interest. In this paper, we will deal with two specific functions.

\subsubsection{Energy}
\label{sec:energy}

In the first  case $V(\bx)$ is proportional to the inverse of the pathloss function between a given user and the nearest base-station.
Hence we have
\begin{equation}\label{eq:power_def}
V_p(\br) = P(\br)  =\gamma r^\beta
\end{equation}
where $\beta$ is the pathloss exponent and $r$ is the distance between the user and the base-station. Here the integral in \eqref{eq:energy_def} will correspond to the total energy consumption of the mobile over time $T$. For simplicity we treat only a single cell which we take to be square with side $R$ and the base-station located at the center. This problem may be generalized straightforwardly to a network of bases located e.g. at a square grid. In that case, we would need to replace the distance $r$ above by $r_{min} = \min_{\br_i} \left| \br-\br_i\right|^\beta$ where $\br_i$ is the location of the $i$th base-station.

The above quantity is the power necessary for a user at a distance $r^\beta$ from a base-station to maintain signal to noise ratio equal to $\gamma$ in the presence of unit variance noise, no other interference and no fading. Despite its simplicity it is not difficult to show that this is asymptotically correct in the massive MIMO setting, where a finite number of users $K$ is served with a very large antenna array of $N$ antennas at the base. Indeed, the $N$ dimensional received signal vector $\by$ at the base is
\begin{equation}\label{eq:Massive_MIMO_def}
 \by = \sum_{k=1}^K \bh_k g(\br_k)^{1/2} x_k + \bz
\end{equation}
where $\bh_k$ for $k=1,\ldots,K$ is the $N$-dimensional channel vector for user $k$, with elements assumed for simplicity to be $\sim {\cal CN}(0,N^{-1})$, $g(\br)=|\br|^{-\beta}$ is the corresponding pathloss function and $x_k$ the transmitted signal, with $\bz$ the $N$-dimensional noise vector with elements $\sim {\cal CN}(0,1)$. Then, in the limit $N\gg K$ the SINR for each user becomes \cite{Tulino2004_RMT_review} asymptotically
\begin{equation}\label{eq:Massive_MIMO_SINR}
 SINR_k = g(\br_k) \ex\left[|x_k|^2\right]
\end{equation}
Requiring this to be equal to $\gamma$ results to \eqref{eq:power_def} after we realize that when the user switches base stations when it reaches the edge of the cell. The pathloss function inversion can easily be implemented through the periodic feedback of a channel quality indication (CQI) to the mobile device, as is usually done. Hence the above power control scheme corresponds to situations where the user needs a constant rate.

\subsubsection{Throughput}
\label{sec:throughput}

A dual uplink transmission strategy to the above for a mobile user in a network corresponds to transmit continuously at a constant power and take advantage of the instances when the channel is good due to proximity to a base-station. In this case the power transmitted is fixed, but it is the communication rate that is fluctuating with distance. This can be expressed as
\begin{equation}\label{eq:capacity_def}
V_c(\br) = C(\br)  =\alpha \min\left(\log\left[1+\frac{p}{r^\beta}\right],R_{max}\right)
\end{equation}
where $r$ is defined as above, $R_{max}=\log(1+p/r_0^\beta)$ is the maximum rate achieved at distance $r_0$ and $\alpha$, $p$ are constants. Here the integral of \eqref{eq:energy_def} will correspond to the total throughput uploaded over time $t$ in bits. The above expression has two interpretations, depending on the context. In the case of a single base antenna, single user case, it corresponds to the outage capacity at location $r_{min}$. In this case, $\alpha=1-p_{out}$ is the probability of non-outage, while $p=-P\log(1-p_{out})$ and $\int dt C(\br(t))$ will correspond to the total goodput. In a massive MIMO multi-user setting, $\alpha=1$ and $p$ is the signal to noise ratio.

\subsection{Mobility Model}
\label{sec:MobilityModel}

We will now specify the dynamics of user mobility. In particular, we assume that the user of interest moves according to a continuous time Markov process. The infinitesimal generator of the process is denoted by the operator $\bM_0$ acting on the space of square integrable functions $\ell^2(\R^2)$. Hence the probability that a user is at location $\br$ at time $t$, given that he was at location $\br_0$ at time $t=0$ can be expressed in terms of $\bM_0$  as follows\cite{Simon1979_FunctionalIntegrationQuantumPhysics_book}
\begin{equation}\label{eq:Markov_definition}
  \prob(\br, t; \br_0, 0) = e^{-\bM_0t}(\br, \br_0)
\end{equation}
where the right hand side is the $\br$, $\br_0$-matrix element of the exponential operator. For concreteness, we will only assume the user performs the simplest Markov process, namely a Brownian motion. This is known to be a good approximation for the long time, large distance properties of a Markov process with finite step variance \cite{Bouchaud_book_FinancialRiskDerivativePricing}. In this case, the infinitesimal generator is simply
\begin{equation}\label{eq:M0_def}
\bM_0 = -\frac{D}{2} \nabla^2
\end{equation}
i.e. proportional to the Laplacian operator, with proportionality constant $D$, the diffusion constant. In this case, the above probability satisfies the diffusion equation
\begin{equation}\label{eq:diffusion_def}
  \frac{\partial \prob(\br, t; \br_0, 0)}{\partial t} = \frac{D}{2} \nabla^2 \prob(\br, t; \br_0, 0)
\end{equation}
with initial condition   $\prob(\br, 0; \br_0, 0) = \delta(\br-\br_0)$, where $\delta(\bx)$ is the two dimensional Dirac $\delta$-function.

We will also need to specify the boundary conditions of the Brownian motion. Specifically, we will assume periodic boundary conditions. This means that the mobile user, when he moves through the boundary of the cell, re-appears from the other side with the same direction. While it is not particularly realistic to assume such a behavior from a user, it is not hard to see that this is a good way to mimic the hand-over to a neighboring cell region. This will be discussed further elsewhere.

%

\section{Mathematical Framework}
\label{sec:math_formalism}

In this section we will prove the main result of the paper, which provides the long time behavior of the metrics introduced above. Then we will obtain analytic limiting behavior for small and large values of the energy and the throughput, as well as the behavior close to the mean.

\begin{theorem} \label{thm:MainResult}
Let $V(\br)\geq 0$ be a continuous, upper bounded function of the distance $\br$ and $E_T$ be given by \eqref{eq:energy_def}, in which the time-dependence of the position $\br(t)$ is due to a Brownian motion on a square $\B=(-R/2,R/2)^2$ with diffusion constant $D$, periodic boundary conditions and initial condition $\br_0=\br(0)$. Then if $A\subset \R$ we have
\begin{eqnarray}\label{eq:thm_1}
\lim_{T\to\infty} \frac{1}{T}\log\prob(E_T/T\in A) &=&  - \inf_{x\in A} I(x)\\
  \label{eq:thmI(x)}
I(x) &=& -\inf_{\lambda\in \R} \{\lambda x-\epsilon_0(\lambda)\}
\end{eqnarray}
where
$\epsilon_0(\lambda)$ is the minimum eigenvalue of the operator
\begin{equation}\label{eq:thm_M_def}
  \bM_\lambda = -\frac{D}{2}\nabla^2 + \lambda V(\br)
\end{equation}
\end{theorem}
\begin{proof}
We now sketch the basic steps of the proof. We start by applying Cramer's theorem\cite{Dembo_book_LargeDeviationsTechniques} to find that in the limit of large $T$ $\log\prob(E_T\in A)$ obeys a large deviation principle with rate function $I(x)$ so that
\begin{eqnarray}
&&  \lim_{T\to\infty} \frac{1}{T} \log Prob(E_T/T\in A) = - \inf_{x\in A} I(x)\\
  \label{eq:I(x)}
&& I(x) = -\inf_{\lambda\in \R} \{\lambda x+ \Lambda(\lambda)\} \\
&& \Lambda(\lambda) = \lim_{T\to \infty} \frac{1}{T}\log\ex_{\br_0}\left[e^{-\lambda E_T}\right] \label{eq:Lambda_def}
\end{eqnarray}
In the above, the expectation is over Brownian paths (random walks) with initial condition $\br(0)=\br_0$. Also, note that the second line has an $\inf$ rather than a $\sup$ as is customary, due to the fact that we have $\Lambda(\lambda)$ with a negative sign in the exponent.

Now, the main trick in this proof is to take advantage of a famous result, namely the Feynman-Kac (FK) formula\cite{Pastur_book1992_SpectraRandomOperators, Simon1979_FunctionalIntegrationQuantumPhysics_book} which states that
\begin{eqnarray}
e^{-\bM_\lambda T}(\br_0, \br_T) =  \prob(\br_T,T; \br_0,0) \ex_{\br_0,\br_T}\left[e^{-\lambda E_T}\right]
\end{eqnarray}
where the right hand side is an expectation over all Brownian motions starting at $r_0$ and ending at $\br_T$ at time $T$.  The operator $\bM_0$ is the infinitesimal generator of the semigroup corresponding to the mobility process and in the case of the simple Brownian motion $\bM_0$ is given in \eqref{eq:M0_def}.  We can relate this equation to the right-hand-side of \eqref{eq:Lambda_def}, we need to integrate the above over $\br_T$, of course with the appropriate probability of the path, namely $\prob(\br_T,T; \br_0,0)$.

The left-hand side of the FK  formula can be expressed very simply using the spectral decomposition of $\bM_\lambda$. Let $\phi_n(\br)$ be the eigenfunctions of $\bM_\lambda$ with corresponding eigenvalue $\epsilon_n(\lambda)$. Then we have
\begin{eqnarray}
e^{-\bM_\lambda T}(\br_0, \br) & =& \sum_{n=0}^\infty \phi_n(\br_0)\phi_n(\br)\, e^{-\epsilon_n(\lambda) T}
\end{eqnarray}
The periodic boundary conditions imposed above mean that the eigenfunctions $\phi_n(\br)$ and their derivatives $\nabla \phi_n(\br)$ have to be continuous on opposite boundaries, i.e. $(x,-R/2)\rightarrow (x,R/2)$ and $(-R/2,y)\rightarrow (R/2,y)$. Integrating over the final position $\br_T$, we obtain
\begin{eqnarray}
 \ex_{\br_0}\left[e^{-\lambda E_T}\right] & =& \sum_{n=0}^\infty \theta_n \phi_n(\br_0)\, e^{-\epsilon_n(\lambda) T} \\ \nonumber
\theta_n & =& \int_{\br\in \B} \phi_n(\br) d\br
\end{eqnarray}
As a result, in the large $T$ limit we have
\begin{eqnarray}
 \ex_{\br_0}\left[e^{-\lambda E_T}\right] \approx \theta_0 \phi_0(\br_0)\, e^{-\epsilon_0(\lambda) T}
\end{eqnarray}
Combining this equation together with \eqref{eq:I(x)}-\eqref{eq:Lambda_def} completes the proof.
\end{proof}
\begin{remark*}
It should be pointed out that there are analogous (but non-local) expressions for generators of stable processes as well\cite{Donsker1975_AsymptoticsWienerSausage}. Also, discrete analogues of the Laplacian can also be found, corresponding to discrete space (continuous time) Markov processes. The FK formula essentially finds the right way to weight the dynamics of the user for which all locations in the cell are equal and the weighting of $V$, which is different as a function of $\br$.
\end{remark*}

We thus see that in the large $T$ limit, we are left with the technical task of finding the minimum eigenvalue $\epsilon_0(\lambda)$ of the operator $\bM_\lambda$ for all $\lambda$. We then plug these into \eqref{eq:I(x)} and solve for $\lambda$. Both steps can be done numerically with not too much effort. However, as we shall see now, we can obtain limiting results for the tails of the distribution analytically.

\begin{figure*}[t]
\centering
\subfigure
[Rate function for energy $I_e\left(\frac{E_T}{T}\right)$]
{\label{fig:rate_fun_energy}
\includegraphics[width=\columnwidth]{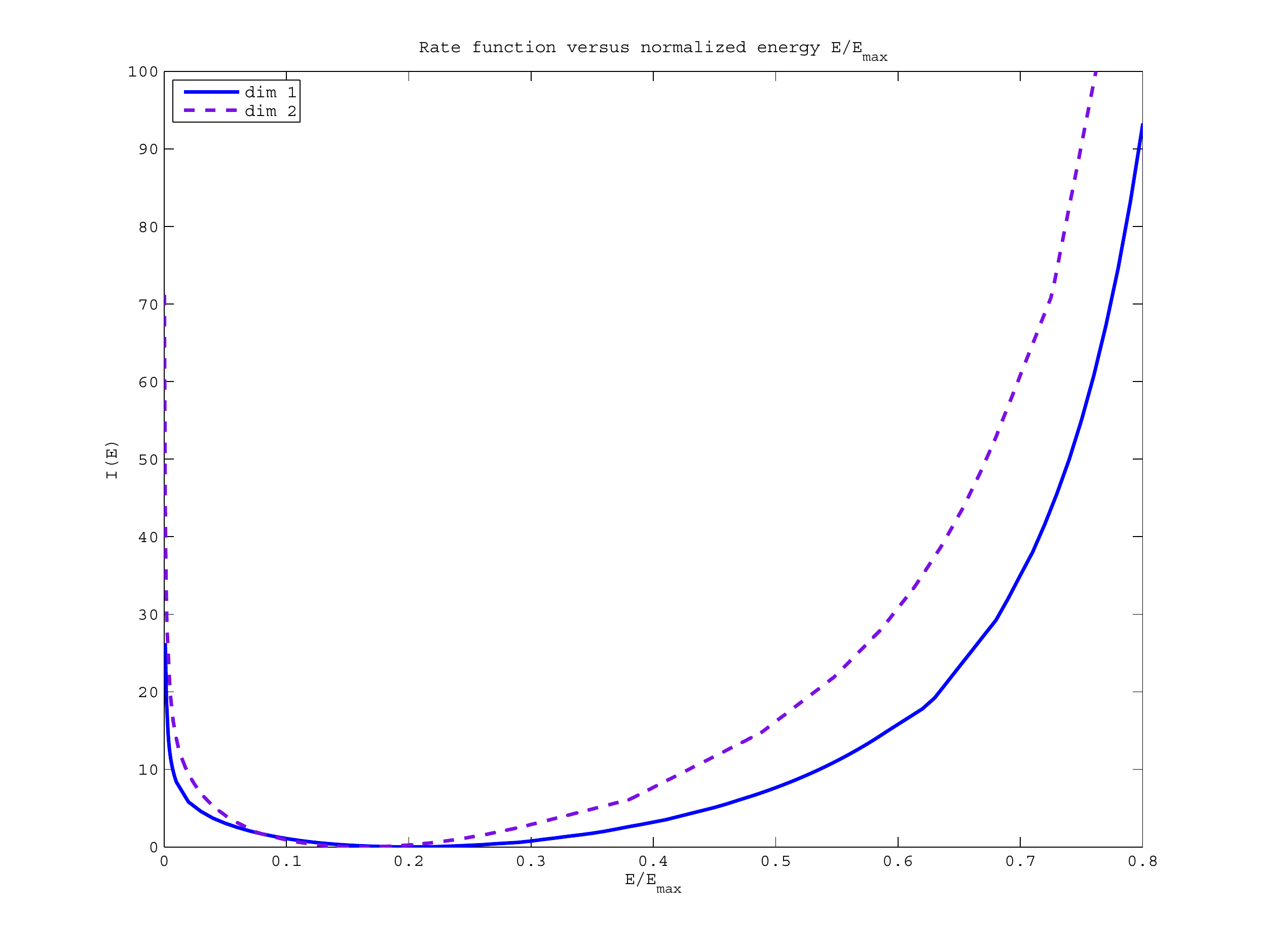}}
\hfill
\subfigure
[Rate function of average throughput $I_b\left(\frac{B_T}{T}\right)$]
{\label{fig:rate_fun_bits}
\includegraphics[width=\columnwidth]{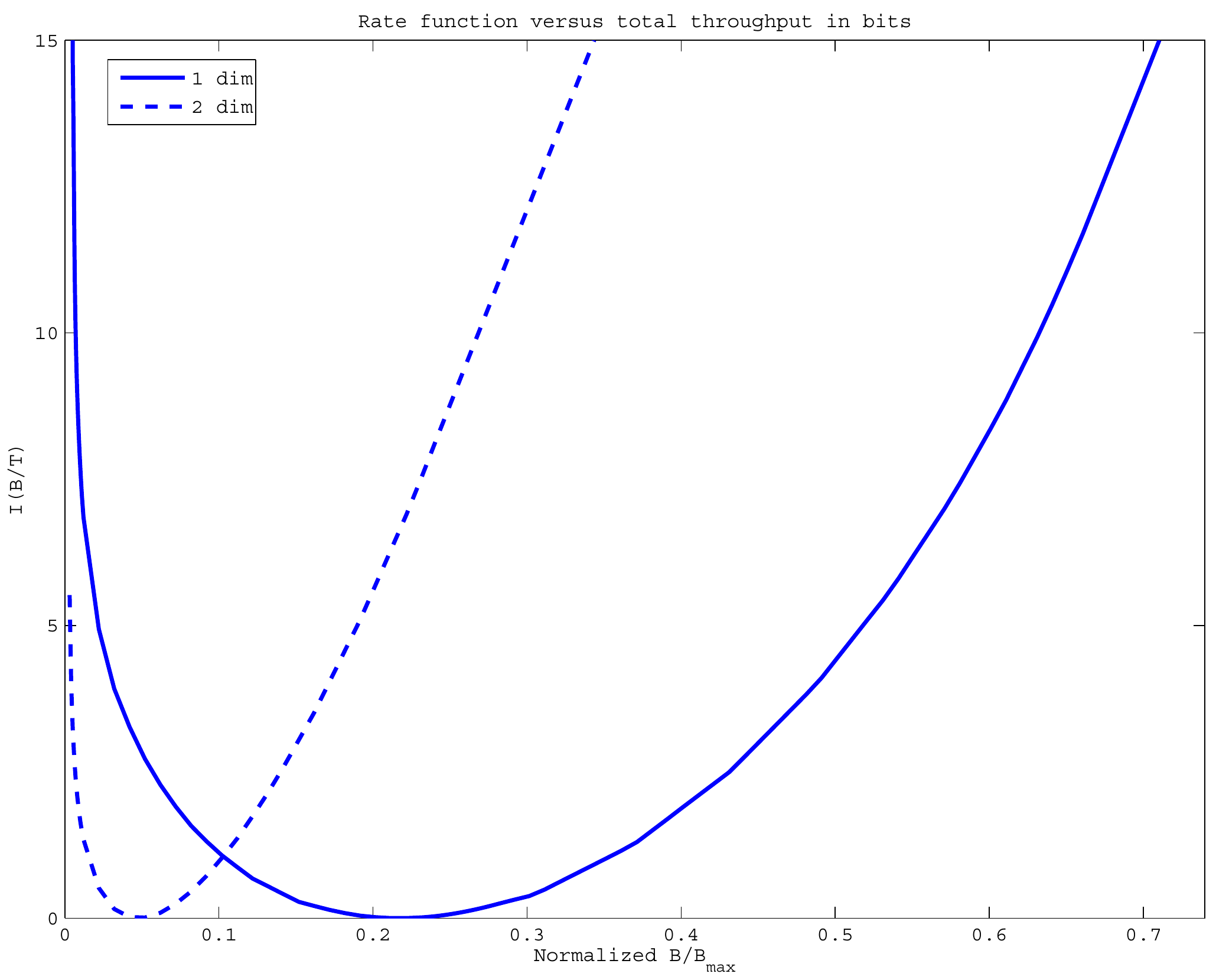}}
\caption{In Figure (a) we plot the rate function for the distribution of the energy of the mobile user. The rate function is plotted for both a one- and two-dimensional square cell. To fit both curves in the same plot the $x$-axis indicates the energy value normalized to its maximal value. Although close to the point where the rate function vanishes, which occurs at the mean value of the energy, the shape is quadratic, its distinctly non-Gaussian shape becomes evident beyond that region. The diverging behavior at the two boundary values, as analyzed in Section \ref{sec:EnergyStatistics} are visible in the plot. In Figure (b) the rate function for the throughput distribution is plotted. The values $\beta=4$, $p=0.01$, $r_0=0.1$ and $\alpha=1$ are used. The very different behavior between one- and two-dimensional rate values for large values of $B$ are due to the lower density of points for $r<r_0$ for two- versus one-dimensional cells.}
\vspace{-1ex}
\end{figure*}

\section{Energy Statistics}
\label{sec:EnergyStatistics}

In this section we will discuss the results for the energy statistics. In this case the quantity of interest is
\begin{equation}\label{eq:energy1}
  E_T = \gamma \int_0^T r(t)^\beta dt
\end{equation}
Plugging this into the methodology above we may obtain the rate function $I_e(x)$, which provides the leading (exponential) term in the distribution of $E_T$ for large $T$.  To do this we will need to find the minimum eigenvalue of the operator in \eqref{eq:thm_M_def} for $V(\br) = \gamma r^\beta$.

We will start by deriving limiting analytic results for large and small values of $E_T$. Starting slightly backwards, we will derive limiting results for the case when $\lambda$ is large and positive. This will correspond for the occurrence probability of unusually small values of the energy. In this case, the minimum eigenvalue $\epsilon_0(\lambda)$ will have an eigenfunction localized close to the bottom of $V(\br)$, namely the center of the cell. As a first approximation, we can forget the boundary of the cell and extend it to infinity. We can therefore rescale the distance variable to $\br  = \overline{r} \bx$, where $\overline{r} = \left(D\lambda^{-1}\gamma^{-1}\right)^{1/(\beta+2)}$ and eliminate the dependence on $\lambda$. The resulting minimum eigenvalue becomes approximately equal to
\begin{equation}\label{eq:min_eig_lambda>0}
  \epsilon(\lambda)\approx D^{\frac{\beta}{\beta+2}} \gamma^{\frac{2}{\beta+2}} \lambda^{\frac{2}{\beta+2}} \varepsilon_0
\end{equation}
where $\varepsilon_0$ is the minimum energy of $\bM_\lambda$ in $\R^2$ with $\lambda=D=1$. After some algebra we obtain
\begin{equation}\label{eq:I(x)_low_energy}
  I_e(x) \approx \frac{D}{R^2} \left\{\frac{\beta}{2}\left(\frac{2\varepsilon_0}{\beta+2}\right)^{1+\frac{2}{\beta}}\right\} \left(\frac{x}{\gamma R^\beta}\right)^{-\frac{2}{\beta}}
\end{equation}
which is valid when $E_T\ll \gamma R^\beta T$.

To obtain the more interesting tails for $E_T\gg P_{mean} T$, we need to analyze the case for large negative $\lambda$. We now see that the minimum of $\lambda V(\br)$ is at the corners of the cell. We now need to do expand the term $|\br|^\beta$ around the value $r_{max}^\beta$, $r_{max}=R/\sqrt{2}$. As a result and  after a shift and $45^o$ rotation of axes, we have
\begin{equation}\label{eq:H_large E}
\bM_\lambda \approx \lambda P_{max} +\bM_0 -\lambda \gamma \beta r_{max}^{\beta-1}\max(|x|,|y|) +O(|\lambda|r_{max}^{\beta-2})
\end{equation}
where $P_{max}=\gamma r_{max}^\beta$. In this case we get
\begin{equation}\label{eq:I(x)_large E}
 I_e(x) \approx \frac{D}{R^2} \frac{4\varepsilon_m^3}{27} \left(1-\frac{x}{P_{max}}\right)^{-2}
\end{equation}
where $\varepsilon_m$ is the minimum eigenvalue of the ``inverted tetrahedron'' operator
\begin{equation}\label{eq:inverted_tetrahedron}
\bM_{eff} = -\frac{1}{2}\nabla^2+\max(|x|,|y|)
\end{equation}
in $L^2(\R^2)$.

Finally, we may obtain the behavior for $E_T\approx T P_{mean}$. With some hindsight, we look in the region of small $|\lambda|$.In this case we treat $\lambda V(\br)$ as small and after performing second order perturbation theory \cite{CohenTannoudji_book} we find that
\begin{equation}\label{eq:gaussian_approx}
\epsilon_0(\lambda)\approx \lambda P_{mean} -\lambda^2\sum_{n\neq 0}\frac{V_{n0}^2}{\epsilon_n}
\end{equation}
where $V_{n0}$ is the expectation of $V(\br)$ over the eigenfunction $\phi_{n0}(\br)$ of the Laplacian and $\epsilon_{n0}$ the corresponding eigenvalue in the square domain. It then turns out that $I(x)$ takes the form
\begin{equation}\label{eq:I(x)gaussian_approx}
  I_e(x)\approx \frac{\left(x-P_{mean}\right)^2}{2\sigma^2}
\end{equation}
where $\sigma^2$ is twice the term multiplying $\lambda^2$ above. This turns out to be the expression for the variance obtained in \cite{Sanguinetti2013_OptimalLinearPrecoding_ICASSP}.

In Fig. \ref{fig:rate_fun_energy}(a) we plot the numerically generated rate function of the energy $I_e(x)$ by calculating numerically the minimum eigenvalue $\epsilon_0(\lambda)$ of \eqref{eq:thm_M_def} with $V(r)=\lambda \gamma r^\beta$ in a square  of unit length for various values of $\lambda$. Then for any given value of $x$, we use this function to find the minimum of $\lambda x-\epsilon_0(\lambda)$. This minimum value is  plotted in the figure. We also did the same for a one-dimensional by finding the minimum eigenvalue the same operator on unit length line. We see that $I_e(x)$ vanishes when $x=P_{mean}=\ex_\br [P(\br)]$ where the expectation is over the cell.. The rate function plotted there provides information about the distribution of the energy $E_T$, as was discussed in \eqref{eq:thmI(x)}. Indeed for $E_T>P_{mean} T$ the probability distribution of $E_T$ is, to logarithmic accuracy equal to
\begin{equation}\label{eq:Prob_energy_E_T>P_mean}
  \prob(E_T>x T)\sim e^{-I_e(x)T}
\end{equation}
A similar expression holds for $E_T<P_{mean} T$.

\section{Throughput Statistics}
\label{sec:Throughput Statistics}

In this section we will present results for the statistics of user throughput in the uplink. As discussed earlier, the appropriate functional here is the integrated rate given by $B_T =\int_0^T C(\br(t)) dt$. As in the previous section we can use this to obtain the rate function $I_b(x)$, which provides the  leading (exponential) term in the distribution of the total bits transmitted $B_T$ for large $T$.

In Fig. \ref{fig:rate_fun_energy}(b) we plot the numerically generated rate function of the throughtput $I_b(B_T/T)$ evaluated similarly as in the previous section for the energy. The larger difference between one- and two-dimensional values of the rate for larger values of the throughput are due to the lower density close to the center in two-versus one-dimensional geometries. As also discussed above the rate function $I_b(x)$ provides to logarithmic accuracy the probability that $B_T<xT$,  thus providing a metric for total throughput ``outage'' probability. Indeed when  $B_T<C_{mean} T$, where $C_{mean}=\ex_\br C(\br)]$ is the value of $x=B_T/T$ where $I_b(x)=0$,  the probability distribution of $B_T$ is, to logarithmic accuracy equal to
\begin{equation}\label{eq:Prob_B_T<C_mean}
  \prob(B_T>x T)\sim e^{-I_b(x)T}
\end{equation}

\section{Discussion and Outlook}
\label{sec:discussion}

Summarizing the contributions of this paper, we have provided a technique to obtain the probability distribution of performance metrics, such as the total throughput and the consumed energy over time for wireless systems by exploiting the statistics of mobility. This methodology can help improve network design and dimensioning, by providing analytic results for low probability events of high energy consumption and/or throughput irregularities. As specific examples, we calculated the long time distribution of the consumed uplink power  and the corresponding total throughput of a single user in a massive MIMO cell, assuming that the user moves according to a Brownian motion.

It should be briefly mentioned that we can generalize the above discussion for Levy random Markov processes that have infinite variance of each step, corresponding to long tailed distributions\cite{Rhee2011_LevyHumanMobility, Scafetta2011_LevyWalkHumanMobility}.  The only difference in the above discussion is the choice of an appropriate infinitesimal generator of the process, which will now be a symmetric stable law of index $\alpha<2$ \cite{Donsker1979_NumberDistinctSitesRW}. In this case $\bM_0$ does not have a local representation (as a derivative), but is still well defined \cite{Donsker1975_AsymptoticsWienerSausage, Donsker1979_NumberDistinctSitesRW, Widom1963_ExtremeEIgenvaluesConvolutionOperators}.
Using the same arguments as above, we can obtain the long term statistics for this mobility model as well. While we leave the more involved presentation of results for a longer paper it is worth mentioning briefly how the above results are expected to change. Focusing for brevity only on the energy case, since $\bM_0$ will have scale dimensions of $R^{-\alpha}$, we can rescale the equations once again and find that for large positive $\lambda$ the minimum eigenvalue will be $\epsilon_0(\lambda)\sim \lambda^{\alpha/(\alpha+\beta)}$. Conversely, for large negative $\lambda$ the minimum eigenvalue will be $\epsilon_0(\lambda)-\lambda r_{max}^\beta\sim \lambda^{\alpha/(1+\alpha)}$. Putting these together which then correspond to a rate function
\begin{eqnarray}\label{eq:I(x)_alpha_small E}
 I_{e,low}(x,\alpha) &\sim&  x^{-\alpha/\beta} \\
 \label{eq:I(x)_alpha_large E}
 I_{e, high}(x,\alpha) &\sim& \left(1-\frac{x}{r_{max}^\beta}\right)^{-\alpha}
\end{eqnarray}

We can also generalize the above results to larger systems with many base-stations. In such situations a user switches between base-stations when crossing the cell boundary in which case the energy consumption in the uplink continues to increase, or in the downlink the power associated to that user is switched off. Mathematically, this has the effect of having a periodic power function, when the cells are assumed to appear in an ordered fashion. Another obvious generalization has to do with taking into account orthogonal (such as OFDMA) channels and treating the total downlink sum-throughput and/or power comsumption. In this system, we end up with an operator $\bM$ describing multiple Brownian motions.

\footnotesize
\bibliographystyle{IEEEtran}
\bibliography{IEEEabrv,C:/Users/ARISLM/ALMDocuments/Dropbox/Work/CurrentWork/bibliography/wireless}

\end{document}